\documentclass[journal]{IEEEtran}
\ifCLASSINFOpdf
\else
\fi
\usepackage{bm}  
\usepackage{multicol}
\usepackage{multirow}
\usepackage{placeins}
\usepackage{graphicx}
\usepackage{amsmath}
\usepackage{array}
\newcolumntype{H}{>{\setbox0=\hbox\bgroup}c<{\egroup}@{}}
\usepackage{graphicx}
\usepackage{caption}
\usepackage{subcaption}
\usepackage{float}
\usepackage{tabularx}
\usepackage{arydshln}
\usepackage[autostyle]{csquotes}
\usepackage{color}

\usepackage[utf8]{inputenc}
\usepackage[english]{babel}

\usepackage{amsthm}

\usepackage{hyperref}
\usepackage{algorithm}
\usepackage{algpseudocode}

\usepackage{chngpage}
\usepackage[compress]{cite}

\newtheorem{theorem}{Theorem}

\def\blackslug{\hbox{\kern1pt\vrule height6pt width4pt  depth1pt\kern1pt}}

\def\edp{\penalty 500\hbox{\quad\blackslug}\ifmmode\else\par
	\vskip4.5pt plus3pt minus2pt\fi}

\newtheorem{reduction test}{Reduction Test}

\newtheorem{definition}{Definition}

\begin{document}
%

%
\title{Solving group Steiner problems as Steiner problems: the rigorous proof}
\author{\IEEEauthorblockN{Yahui Sun} \\
	\IEEEauthorblockA{\url{https://yahuisun.com}}
}
%
	\maketitle
\begin{abstract}
The Steiner tree problems are well-known NP-hard problems that have diverse applications. Duin et al. (2004) have intuitively proposed the widely-used transformation from the classical group Steiner tree problem to the classical Steiner tree problem in graphs. This transformation has not been rigorously proven so far. Specifically, the large M value that is used in this transformation has not been specified. In this paper, we address this issue by rigorously prove this transformation for a specific large M value.
\end{abstract}

	
	%
\section{main content}
	
The Steiner tree problems, which are named after Jakob Steiner, a 19th-century mathematician at the University of Berlin, have diverse applications to social, biomedical, and computer communication networks \cite{cpan2018}. The classical Steiner Tree Problem in Graphs (STPG) \cite{tspi1971} is about finding the minimum-cost subgraph to connect some compulsory vertices together in a connected undirected graph with positive edge costs. Many more complex Steiner tree problems in graphs have been developed based on it, including the classical Group Steiner Tree Problem (GSTP) \cite{csta1999}, where Steiner trees must contain at least one vertex in each group of vertices. Duin et al. (2004) \cite{sgsp2004} intuitively showed that GSTP can be transformed to STPG by adding dummy vertices and edges. Even though this transformation has been widely used in the last few years (e.g. \cite{fato2009}), it has not been rigorously proven so far. Specifically, the large M value that is used in this transformation has not been specified. In this paper, we address this issue by rigorously prove this transformation for a specific large M value.

First, we formally define STPG and GSTP as follows.

\begin{definition} [The classical Steiner Tree Problem in Graphs] 
	Let $G(V,E,T,c)$ be a connected undirected graph, where $V$ is the set of vertices, $E$ is the set of edges, $T$ is a subset of $V$ that we refer to as compulsory vertices, and $c$ is a function which maps each edge in $E$ to a positive value that we refer to as edge cost. The purpose is to find a connected subgraph $G'(V',E'), T \subseteq V'\subseteq V, E'\subseteq E $  with the minimum cost $c(G')=\sum_{e \in E'}{c(e)}$.
\end{definition}
	
\begin{definition} [The classical Group Steiner Tree Problem] 
	Let $G(V,E,\Gamma,c)$ be a connected undirected graph, where $V$ is the set of vertices, $E$ is the set of edges, $\Gamma$ is a collection of subsets of $V$ that we refer to as groups, and $c$ is a function which maps each edge in $E$ to a positive value that we refer to as edge cost. The purpose is to find a connected subgraph $G'(V',E'), V'\subseteq V, E'\subseteq E $ with the minimum cost $c(G')=\sum_{e \in E'}{c(e)}$, and for each group $g \in \Gamma$, $g \cap V' \neq \emptyset$. 
\end{definition}

We refer to the optimal solutions to GSTP and STPG as Group Steiner Minimum Tree (GSMT) and Steiner Minimum Tree (SMT) respectively.  Then, we rigorously prove Duin et al.'s transformation from GSTP to STPG \cite{sgsp2004} as follows. 

\begin{theorem} \label{Theorem: transformation}
	Let $G(V,E,\Gamma,c)$ be a connected undirected graph. For each group $g \in \Gamma$, add a compulsory vertex $v_g$ and edges $(v_g,j)|\forall j \in g$ such that $c(v_g,j)=M= \sum_{e \in E}{c(e)}$. Let $\Theta$ be the GSMT on $G$, and $\Theta'$ be the SMT on the new graph $G'$. Then $\Theta=\Theta' \setminus \sum (v_g,j)$.
\end{theorem}

\begin{proof}
	Clearly, there is a feasible solution to STPG on  $G'$ whose cost is $c(\Theta) +M|\Gamma| \geq c(\Theta')$. Let $\Theta_1$ be another tree in $G'$ that contains all the new vertices, and every new vertex is a leaf. Suppose that there is a new vertex $v_g$ in $\Theta'$ that is not a leaf, then
	\begin{eqnarray}\label{}
	\begin{array}{r l}
	c(\Theta') \geq &  c(\Theta' \setminus \sum (v_g,j)) +M(|\Gamma|+1) \geq\\
	c(\Theta_1) =  &  c(\Theta_1 \setminus \sum (v_g,j)) +M|\Gamma|
	\end{array}
	\end{eqnarray} 
	which is not possible. Thus, every new vertex is a leaf in $\Theta'$, which means that 1) $c(\Theta) \geq c(\Theta') - M|\Gamma| = c(\Theta' \setminus \sum (v_g,j))$; and 2) $\Theta' \setminus \sum (v_g,j)$ is a connected tree and thus a feasible solution to GSTP on $G$, i.e., $c(\Theta) \leq  c(\Theta' \setminus \sum (v_g,j))$. Therefore, $c(\Theta) = c(\Theta' \setminus \sum (v_g,j))$. This theorem holds.
\end{proof}

	\bibliographystyle{ieeetr}
	\bibliography{YahuiBibIEEE}

\end{document}